\documentclass[%
 aip,
 sd,%
 amsmath,amssymb,
]{revtex4-1}

\usepackage{dcolumn}
\usepackage{bm}
\usepackage{amsthm,mathrsfs}
\usepackage{color}
\newcommand{\mm}[1]{#1}

\newcommand{\Rm}{\mathbb{R}}
\newcommand{\Cm}{\mathbb{C}}

\newcommand{\Sm}{\mathbb{S}}
\newcommand{\ba}{\begin{eqnarray*}}
\newcommand{\ea}{\end{eqnarray*}}
\newcommand{\be}{\begin{equation}}
\newcommand{\ee}{\end{equation}}
\newcommand{\bea}{\begin{eqnarray}}
\newcommand{\eea}{\end{eqnarray}}
\newcommand{\va}{\varphi}
\newcommand{\pp}{\partial}
\newcommand{\vv}[1]{\boldsymbol{\mathrm{#1}}}
\newcommand{\hvv}[1]{\boldsymbol{\hat{\mathrm{#1}}}}
\newcommand{\uv}{\boldsymbol{{\hat{\Omega}}}}
\newcommand{\uvk}{\boldsymbol{\hat{\mathrm{k}}}}

\newcommand{\pint}{\:\mathcal{P}\!\!\int}
\newcommand{\rrf}[1]{\mathop{\mathcal{R}_{{#1}}}}
\newcommand{\irrf}[1]{\mathop{\mathcal{R}_{{#1}}^{-1}}}

\newtheorem{thm}{Theorem}[section]
\newtheorem{lem}[thm]{Lemma}

\newtheorem{prop}[thm]{Proposition}
\newtheorem{defn}[thm]{Definition}
\theoremstyle{remark}\newtheorem{rmk}[thm]{Remark}

\begin{document}

\preprint{arXiv:1511.05723}

\title[]{
The radiative transport equation in flatland with separation of variables}

\author{Manabu Machida}
\email{mmachida@ms.u-tokyo.ac.jp}
\altaffiliation[]{Department of Mathematical Sciences, 
The University of Tokyo, Komaba, Meguro, Tokyo 153-8914, Japan}
\altaffiliation{\mm{Current address: 
Institute for Medical Photonics Research, 
Hamamatsu University School of Medicine, 
Hamamatsu, Shizuoka 431-3192, Japan}
}

\date{\today}

\begin{abstract}
The linear Boltzmann equation can be solved with separation of variables
in one dimension, i.e., in three-dimensional space with planar symmetry.
In this method, solutions are given by superpositions of eigenmodes which
are sometimes called singular eigenfunctions. In this paper, we explore
the singular-eigenfunction approach in flatland or two-dimensional space.
\end{abstract}

\maketitle

\section{Introduction}

We consider the radiative transport equation or linear Boltzmann equation in 
flatland or in two spatial dimensions. The Green's function 
$G(\vv{\rho},\va;\va_0)$ ($\vv{\rho}={^t}(x,y)\in\Rm^2$, $0\le\va\mm{<}2\pi$, 
$0\le\va_0\mm{<}2\pi$) satisfies
\be
\left(\uv\cdot\nabla+1\right)G(\vv{\rho},\va;\va_0)=
\varpi\int_0^{2\pi}p(\va,\va')G(\vv{\rho},\va';\va_0)\,d\va'+
\delta(\vv{\rho})\delta(\va-\va_0),
\label{rte2d}
\ee
where $\uv={^t}(\cos\va,\sin\va)$ is a unit vector in $\Sm$, 
$\nabla={^t}(\mm{\frac{\pp}{\pp x}},\mm{\frac{\pp}{\pp y}})$, and 
$\varpi\in(0,1)$ is the albedo for single scattering. We have
\[
G(\vv{\rho},\va;\va_0)\to0\quad\mbox{as}\quad|\vv{\rho}|\to\infty.
\] 
We suppose \mm{that} the scattering phase function 
$p(\va,\va')\in L^{\infty}(\Sm\times\Sm)$ is nonnegative and is 
\mm{normalized as}
\[
\int_0^{2\pi}p(\va,\va')\,d\va'=1.
\]
\mm{We assume that $p(\va,\va')$ is given by}
\ba
p(\va,\va')
&=&
\frac{1}{2\pi}\sum_{m=-L}^L\beta_me^{im(\va-\va')}
\\
&=&
\frac{1}{2\pi}+\frac{1}{\pi}\sum_{m=1}^L\beta_m\cos[m(\va-\va')],
\ea
where $L\ge0$, $\beta_0=1$, \mm{$-1<\beta_l<1$ ($l>0$),} and 
$\beta_{-m}=\beta_m$. We put
\[
\beta_m=0\qquad\mbox{for}\quad |m|>L.
\]
The Henyey-Greenstein model \cite{Henyey-Greenstein41} is obtained by taking 
the limit $L\to\infty$ and putting $\beta_m=\mathrm{g}^{|m|}$ with a constant 
$\mathrm{g}\in(-1,1)$, where 
$\mathrm{g}=\int_0^{2\pi}\cos(\va-\va')p(\va,\va')\,d\va'$.

The radiative transport equation which depends on one spatial variable in 
three dimensions has attracted a lot of attention in linear transport 
theory. The singular-eigenfunction approach was explored as early as 1945 by 
Davison \cite{Davison45}. After further efforts such as Van Kampen 
\cite{VanKampen55}, Davison \cite{Davison57}, and Wigner \cite{Wigner61}, 
Case established a way of 
finding solutions with separation of variables \cite{Case59,Case60}. 
The method, called Case's method \cite{Case-Zweifel,Duderstadt-Martin}, was 
soon extended to anisotropic scattering \cite{McCormick-Kuscer66,Mika61}.

On the other hand, the technique of rotated reference frames has existed 
in transport theory since 1964 \cite{Dede64,Kobayashi77}. This method didn't 
sound promising even though the idea was interesting. However, a decade ago, 
Markel succeeded in constructing an efficient numerical algorithm 
\cite{Markel04}, which is called the method of rotated reference frames 
\cite{Machida10,Panasyuk06,Schotland-Markel07}, to find solutions \mm{of} the 
three-dimensional radiative transport equation by reinventing rotated 
reference frames.

Recently, the above two, separation of variables and rotated reference frames, 
were merged and Case's method was extended to three spatial variables 
\cite{Machida14}. \mm{There, r}otated reference frames provide a tool to 
\mm{reduce three-dimensional equations to one dimension. With this tool, f}or 
example, the $F_N$ method \cite{Siewert78,Siewert-Benoist79} was extended to 
three dimensions \cite{Machida15}.

The radiative transport equation is used in various subfields in science 
and engineering \cite{Apresyan-Kravtsov} such as light propagation in 
biological tissue \cite{Arridge99,Arridge-Schotland09}, clouds, and ocean 
\cite{Sobolev,Thomas-Stamnes}, seismic waves \cite{Sato-Fehler}, light in the 
interstellar medium \cite{Chandrasekhar,Peraiah}, neutron transport 
\cite{Case-Zweifel}, and remote sensing \cite{Ishimaru}. In these cases, 
usually three dimensions are most important. There are, however, cases where 
two dimensions have particular interests. Such flatland transport equations 
appear, for example, in wave scattering in the marginal ice zone 
\cite{Kohout-Meylan06} and wave transport along a surface with random 
impedance \cite{Bal-etal00}. Sometimes optical tomography is considered 
in flatland \cite{Arridge-etal06,GR-Kim15,Heino-Arridge-Sikora-Somersalo03,Klose-Netz-Beuthan-Hielscher02,Tarvainen-Vauhkonen-Arridge08}. 
The two-dimensional transport equation is also used for thermal 
radiative transfer \cite{Johnson-Larsen11} and heat transfer 
\cite{Volokitin-Persson01}. 
We note that the method of rotated reference frames was applied to 
two-dimensional space \cite{LK11,LK12,LK13}.

In this paper, we consider the linear Boltzmann equation or radiative 
transport equation in flatland. Let $\mu$ denote the cosine of $\va$:
\[
\mu=\cos\va,\qquad\va\in[0,2\pi\mm{)}.
\]
Let us introduce polynomials $\gamma_m(z)$ ($z\in\Cm$) which satisfy the 
following three-term recurrence relation.
\be
2\nu h_m\gamma_m(\nu)-\gamma_{m+1}(\nu)-\gamma_{m-1}(\nu)=0,
\label{gammarecurr}
\ee
with initial terms
\[
\gamma_0(\nu)=1,\qquad\gamma_1(\nu)=(1-\varpi)\nu.
\]
Here,
\[
h_m=1-\varpi\beta_m.
\]
We have
\[
\gamma_m(-\nu)=(-1)^m\gamma_m(\nu),\qquad\gamma_{-m}(\nu)=\gamma_m(\nu).
\]
The function $g(z,\va)$ is given by
\be
g(\nu,\va)=1+2\sum_{m=1}^L\beta_m\gamma_m(\nu)\cos{m\va}.
\label{deffuncg}
\ee
We introduce
\be
\Lambda(z)=1-\frac{\varpi z}{2\pi}\int_0^{2\pi}\frac{g(z,\va)}{z-\mu}\,d\va,
\qquad z\in\Cm\setminus[-1,1].
\label{defbigLambda}
\ee
Suppose $\Lambda(z)$ has $M=M(L,\varpi,\beta_m)$ positive roots. Let 
$\nu_j$ ($j=0,\dots,M-1$) be positive roots which satisfy $\Lambda(\nu_j)=0$. 
We further introduce
\be
\lambda(\nu)=1-\frac{\varpi\nu}{2\pi}\pint_0^{2\pi}
\frac{g(\nu,\va)}{\nu-\mu}\,d\va,\qquad\nu\in(-1,1),
\label{deflittlelambda}
\ee
where $\mathcal{P}$ denotes Cauchy's principal value. In the next section, we 
will see that singular eigenfunctions in flatland are obtained as
\be
\phi(\nu,\va)=\frac{\varpi\nu}{2\pi}\mathcal{P}\frac{g(\nu,\va)}{\nu-\mu}
+\mm{\frac{\sqrt{1-\nu^2}}{2}}\lambda(\nu)\delta(\nu-\mu),
\label{defsingulareigenfunctions}
\ee
where $\nu=\pm\nu_j$ ($j=0,\dots,M-1$) or $\nu\in(-1,1)$. Let us introduce 
the normalization factor
\be
\mathcal{N}(\nu)=
\left\{\begin{aligned}
\frac{2\nu}{\sqrt{1-\nu^2}}\left[\left(\frac{\varpi\nu}{2}\right)^2
g(\nu,\va_{\nu})^2+\lambda(\nu)^2\right],
&\quad\nu\in(-1,1),
\\
\left(\frac{\varpi\nu}{2}\right)^2g(\nu,\va_{\nu})\frac{d\Lambda(\nu)}{d\nu},
&\quad\nu\notin[-1,1].
\end{aligned}\right.
\label{normalizationfactor}
\ee
Here,
\be
\va_{\nu}=\left\{\begin{aligned}
\cos^{-1}{(\nu)},&\quad\nu\in[-1,1],
\\
i\cosh^{-1}{(\nu)},&\quad\nu>1,
\\
\pi+i\cosh^{-1}{(|\nu|)},&\quad\nu<-1.
\end{aligned}\right.
\label{defvarphinu}
\ee
We note that $0\le\cos^{-1}{\nu}\le\pi$ for $\nu\in[-1,1]$ and 
$\cosh^{-1}{(|\nu|)}=\ln\left(|\nu|+\sqrt{\nu^2-1}\right)$ for $|\nu|>1$. 
Similarly, we use $\va_{\uvk(\nu q)}$ for the analytically continued angle 
such that
\[
\cos{\va_{\uvk(\nu q)}}=\sqrt{1+(\nu q)^2},\qquad
\sin{\va_{\uvk(\nu q)}}=-i\nu q,
\]
for $\nu,q\in\Rm$. As is shown in Section \ref{2dGreen}, the Green's function 
in flatland is obtained as
\bea
G(\vv{\rho},\va;\va_0)
&=&
\frac{1}{2\pi}\int_{-\infty}^{\infty}e^{iqy}\Biggl[
\sum_{j=0}^{M-1}
\phi(\pm\nu_j,\va-\va_{\uvk(\pm\nu_jq)})
\phi(\pm\nu_j,\va_{\mm{0}}-\va_{\uvk(\pm\nu_jq)})
\nonumber \\
&\times&
\frac{1}{\sqrt{1+(\nu_jq)^2}\mathcal{N}(\nu_j)}
e^{-\sqrt{1+(\nu_jq)^2}|x|/\nu_j}
\nonumber \\
&+&
\int_0^1
\phi(\pm\nu,\va-\va_{\uvk(\pm\nu q)})
\phi(\pm\nu,\va_{\mm{0}}-\va_{\uvk(\pm\nu q)})
\nonumber \\
&\times&
\frac{1}{\sqrt{1+(\nu q)^2}\mathcal{N}(\nu)}
e^{-\sqrt{1+(\nu q)^2}|x|/\nu}\,d\nu
\Biggr]\,dq,
\label{mainresult}
\eea
where upper signs are used for $x>0$ and lower signs are used for $x<0$.

The main purpose of the present paper is to derive (\ref{mainresult}). 
We will first consider the one-dimensional problem in two spatial dimensions 
with separation of variables in Sections \ref{1d} and 
\ref{1dGreen}. Two-dimensional singular eigenfunctions are considered in 
Section \ref{2d}. In particular their orthogonality relations are established. 
Then in Section \ref{2dGreen}, we obtain the Green's function for the 
radiative transport equation in two dimensions by extending the 
one-dimensional problem to two dimensions using rotated reference frames. 
Finally, Section \ref{conclusions} is devoted to concluding remarks. In 
Appendix, the Fourier-transform method is explained as an alternative 
approach.

\section{One-dimensional transport theory in flatland}
\label{1d}

We begin with the one-dimensional homogeneous problem \mm{given by}
\be
\left(\mu\frac{\pp}{\pp x}+1\right)\psi(x,\va)
=\varpi\int_0^{2\pi}p(\va,\va')\psi(x,\va')\,d\va'.
\label{rte1dhomo}
\ee
We assume that solutions are given by the following form of separation of 
variables with separation constant $\nu$\mm{.}
\[
\psi_{\nu}(x,\va)=\phi(\nu,\va)e^{-x/\nu}.
\]
We normalize $\phi(\nu,\va)$ as
\[
\int_0^{2\pi}\phi(\nu,\va)\,d\va=1.
\]
We then have
\be
\left(1-\frac{\mu}{\nu}\right)\phi(\nu,\va)
=\frac{\varpi}{2\pi}g(\nu,\va),
\label{phieq}
\ee
where
\[
g(\nu,\va)=1+
2\sum_{m=1}^L\beta_m\left[\gamma_m(\nu)\cos{m\va}+s_m(\nu)\sin{m\va}\right].
\]
\mm{As is shown below, we have}
\bea
\gamma_m(\nu)&=&\int_0^{2\pi}\phi(\nu,\va)\cos{m\va}\,d\va,
\label{defgammam}
\\
s_m(\nu)&=&\int_0^{2\pi}\phi(\nu,\va)\sin{m\va}\,d\va.
\label{defsm}
\eea
\mm{From (\ref{phieq}), we obtain (\ref{defsingulareigenfunctions}).} Direct 
calculation shows that $\gamma_m(\nu)$ satisfy (\ref{gammarecurr}). 
Since (\ref{rte1dhomo}) implies $\phi(\nu,-\va)=\phi(\nu,\va)$, coefficients 
for $\sin{m\va}$ should be zero. Indeed, $s_m(\nu)=0$ for all $m$ as is shown 
below. For a function $f(\va)\in\Cm$, we have
\ba
\int_0^{2\pi}f(\va)\,d\va
&\mm{=}&
\mm{\int_0^{\pi}\left[f(\va)+f(\va+\pi)\right]\,d\va}
\\
&=&
\int_{-1}^1\frac{f(\cos^{-1}{\mu})+f(2\pi-\cos^{-1}{\mu})}{\sqrt{1-\mu^2}}
\,d\mu
\\
&=&
\int_{-1}^1\frac{f(\cos^{-1}{\mu})+f(\pi+\cos^{-1}{\mu})}{\sqrt{1-\mu^2}}
\,d\mu,
\ea
where we used $\int_0^{\pi}f(\va+\pi)\,d\va=\int_0^{\pi}f(2\pi-\va)\,d\va$. 
\mm{By using $\cos^{-1}(-\nu)=\pi-\cos^{-1}(\nu)$, we note that}
\[
\int_0^{2\pi}\delta(\nu-\mu)\sin(m\va)\,d\va=0,\qquad\mm{\nu\in\Rm.}
\]
\mm{If we plug (\ref{defsingulareigenfunctions}) into (\ref{defsm}),} we obtain
\[
s_m(\nu)=v_m(\nu)+\sum_{n=1}^LB_{mn}(\nu)s_n(\nu),
\]
where
\ba
v_m(\nu)&=&
\frac{\varpi\nu}{2\pi}\pint_0^{2\pi}\frac{\sin{m\va}}{\nu-\mu}\,d\va+
\frac{\varpi\nu}{\pi}\sum_{n=1}^L\beta_n\pint_0^{2\pi}
\frac{\cos(n\va)\sin(m\va)}{\nu-\mu}\,d\va,
\\
B_{mn}(\nu)&=&
\frac{\varpi\nu}{\pi}\beta_n\pint_0^{2\pi}\frac{\sin(n\va)\sin(m\va)}{\nu-\mu}
\,d\va.
\ea
However, \mm{we see that}
\[
\pint_0^{2\pi}\frac{\cos(n\va)\sin(m\va)}{\nu-\mu}\,d\va=
\mm{\pint_0^{2\pi}\frac{\sin[(m+n)\va]+\sin[(m-n)\va]}{\nu-\mu}\,d\va=}
0,
\]
for all $m,n$\mm{. H}ence $v_m(\nu)=0$. Since the $L\times L$ matrix whose 
$mn$-element is given by $\delta_{mn}-B_{mn}(\nu)$ is invertible, we see that 
$s_m(v)=0$. Therefore we obtain (\ref{deffuncg}) and singular eigenfunctions 
(\ref{defsingulareigenfunctions}). We have
\[
g(\nu,\va+2\pi)=g(\nu,\va)=g(\nu,-\va),\quad g(-\nu,\va)=g(\nu,\va+\pi).
\]
\mm{When $\nu\in(-1,1)$, we obtain (\ref{deflittlelambda}) by integrating 
(\ref{defsingulareigenfunctions}) over $\va$.} Note that 
\mm{from (\ref{deflittlelambda})}
\[
\lambda(-\nu)=
\mm{1-\frac{\varpi\nu}{2\pi}\pint_0^{2\pi}\frac{g(-\nu,\va)}{\nu+\mu}\,d\va=
1-\frac{\varpi\nu}{2\pi}\pint_0^{2\pi}\frac{g(\nu,\va)}{\nu-\mu}\,d\va=}
\lambda(\nu).
\]

For $\nu\notin[-1,1]$ we have
\[
1=\int_0^{2\pi}\phi(\nu,\va)\,d\va
=\frac{\varpi\nu}{2\pi}\int_0^{2\pi}\frac{g(\nu,\va)}{\nu-\mu}\,d\va.
\]
Therefore discrete eigenvalues $\nu\in\Cm\setminus[-1,1]$ are roots of the 
function $\Lambda(\nu)$ given in (\ref{defbigLambda}). We note that if 
$\nu$ is a discrete eigenvalue, so is $-\nu$ because 
$\Lambda(-\nu)=\Lambda(\nu)$. That is, the eigenvalues $\pm\nu$ appear in 
pairs.

Singular eigenfunctions satisfy the following relations.
\[
\phi(\nu,\va+2\pi)=\phi(\nu,\va)=\phi(\nu,-\va),\quad
\phi(-\nu,\va)=\phi(\nu,\va+\pi).
\]

\begin{prop}
Discrete eigenvalues are real.
\end{prop}

\begin{proof}
Let $m_B$ be a positive integer such that $m_B\ge L$. We first show that if 
$\nu$ satisfies $\gamma_{m_B+1}(\nu)=0$, then this $\nu$ is an eigenvalue of 
matrix $B$, which is the real symmetric matrix defined below. Hence 
$\nu\in\Rm$. Next we show that zeros of $\gamma_{m_B+1}(\nu)$ become roots of 
$\Lambda(\nu)$ as $m_B\to\infty$. With these two, the proof is completed.

Let us note that 
\[
\left|\beta_m\right|=
\left|\beta_me^{im\va}\right|=
\left|\int_0^{2\pi}p(\va,\va')e^{im\va'}\,d\va'\right|\le
\int_0^{2\pi}p(\va,\va')\,d\va'=1.
\]
Hence $h_m>0$ for all $m$. We can rewrite the three-term 
recurrence relation (\ref{gammarecurr}) as
\[
b_m\left(\sqrt{2h_{m-1}}\gamma_{m-1}\right)+
b_{m+1}\left(\sqrt{2h_{m+1}}\gamma_{m+1}\right)=
\nu\left(\sqrt{2h_m}\gamma_m\right),
\]
where
\[
b_m=\frac{1}{2\sqrt{h_{m-1}h_m}}.
\]
Similar to \cite{Garcia-Siewert89}, we consider a tridiagonal 
$(2m_B+1)\times(2m_B+1)$ matrix $\bar{B}$ whose elements are given by
\ba
\bar{B}_{mm'}
&=&
b_m\delta_{m',m-1}+b_{m+1}\delta_{m',m+1}
\\
&+&
b_{-m_B}\sqrt{\frac{h_{-m_B-1}}{h_{-m_B}}}\frac{\gamma_{-m_B-1}}{\gamma_{-m_B}}
\delta_{m,-m_B}\delta_{m',-m_B}
\\
&+&
b_{m_B+1}\sqrt{\frac{h_{m_B+1}}{h_{m_B}}}\frac{\gamma_{m_B+1}}{\gamma_{m_B}}
\delta_{m,m_B}\delta_{m',m_B},
\ea
for $-m_B\le m\le m_B$ and $-m_B\le m'\le m_B$. Therefore if $\nu$ is a zero 
of $\gamma_{m_B+1}(\nu)$, we see that $\nu$ is an eigenvalue of the matrix $B$ 
whose elements are given by
\[
B_{mm'}=b_m\delta_{m',m-1}+b_{m+1}\delta_{m',m+1},
\]
for $-m_B\le m,m'\le m_B$. Since $B$ is real symmetric, $\nu$ is real. 
In particular we can say that $\nu\in\Rm$ even in the limit $m_B\to\infty$ 
\cite{Case74,Shultis-Hill76}.

Next we will explore the connection between roots of $\Lambda$ and 
$\gamma_{m_B+1}$ by repeatedly using the Christoffel-Darboux formula 
\cite{Garcia-Siewert82,Inonu70}.

In addition to (\ref{gammarecurr}), we introduce
\be
2zp_m(z)-p_{m+1}(z)-p_{m-1}(z)=0,
\label{precurr}
\ee
with
\[
p_0(z)=1,\qquad p_{\pm1}=z.
\]
Furthermore we define
\be
P_m(z)=\frac{1}{2\pi}\int_0^{2\pi}\frac{\cos{m\va}}{z-\mu}\,d\va.
\label{bigPdef}
\ee
We have
\be
2zP_m(z)-P_{m+1}(z)-P_{m-1}(z)=2\delta_{m0},
\label{bigPrecurr}
\ee
with
\[
P_1(z)=zP_0(z)-1.
\]
Direct calculation of $\Lambda(z)$ in (\ref{defbigLambda}) shows
\[
\Lambda(z)=1-\varpi z\left[P_0(z)+2\sum_{m=1}^L\beta_m\gamma_m(z)P_m(z)\right].
\]
Let us consider 
$P_m(z)\times(\ref{gammarecurr})-\gamma_m(z)\times(\ref{bigPrecurr})$. We have
\ba
-2\varpi z\beta_mP_m(z)\gamma_m(z)
&=&
\Bigl(P_m(z)\gamma_{m+1}(z)-P_{m+1}(z)\gamma_m(z)\Bigr)
\\
&-&
\Bigl(P_{m-1}(z)\gamma_m(z)-P_m(z)\gamma_{m-1}(z)\Bigr)
\\
&-&
2\delta_{m0}.
\ea
We then take the summation on both sides by $\sum_{m=1}^{m_B}$. As a result we 
obtain
\be
\Lambda(z)=P_{m_B}(z)\gamma_{m_B+1}(z)-P_{m_B+1}(z)\gamma_{m_B}(z).
\label{CDeq1}
\ee
Next, $\sum_{m=1}^{m_B}\left[
p_m(z)\times(\ref{bigPrecurr})-P_m(z)\times(\ref{precurr})\right]$ yields
\be
P_{m_B}(z)p_{m_B+1}(z)-P_{m_B+1}(z)p_{m_B}(z)=1.
\label{CDeq2}
\ee
Moreover we obtain by $\sum_{m=1}^{m_B}\left[
\gamma_m(z)\times(\ref{precurr})-p_m(z)\times(\ref{gammarecurr})\right]$, 
\be
\varpi zg(z,\va_z)=\gamma_{m_B}(z)p_{m_B+1}(z)-\gamma_{m_B+1}(z)p_{m_B}(z).
\label{CDeq3}
\ee
By using (\ref{CDeq1}), (\ref{CDeq2}), and (\ref{CDeq3}), we obtain
\begin{widetext}
\ba
p_{m_B+1}(z)\Lambda(z)
&=&
p_{m_B+1}(z)P_{m_B}(z)\gamma_{m_B+1}(z)-p_{m_B+1}(z)P_{m_B+1}(z)\gamma_{m_B}(z)
\\
&=&
\gamma_{m_B}(z)+\left[p_{m_B}(z)\gamma_{m_B+1}(z)-p_{m_B+1}(z)\gamma_{m_B}(z)
\right]P_{m_B+1}(z)
\\
&=&
\gamma_{m_B+1}(z)-\varpi zg(z,\va_z)P_{m_B+1}(z).
\ea
\end{widetext}
Therefore we obtain
\[
\Lambda(z)=\frac{\gamma_{m_B+1}(z)}{p_{m_B+1}(z)}
-\varpi zg(z,\va_z)\frac{P_{m_B+1}(z)}{p_{m_B+1}(z)}.
\]
However, $P_{m_B+1}(z)$ vanishes as $m_B\to\infty$ due to 
the Riemann-Lebesgue lemma. Thus discrete eigenvalues are zeros of 
$\gamma_{m_B+1}$ as $m_B\to\infty$.

\end{proof}

We suppose there are $2M$ discrete eigenvalues 
$\pm\nu_j$ ($j=0,\dots,M-1$) such that $\nu_j>1$ and $\Lambda(\pm\nu_j)=0$.

\begin{defn}
Let $\sigma$ denote the set of ``eigenvalues''.
\[
\sigma=\{\nu\in\Rm;\;\nu\in(-1,1)\;\mbox{or}\;\nu=\pm\nu_j,\;j=0,1,\dots,M-1\}.
\]
\end{defn}

For later calculations, we will prepare some notations. 
Let $\va_z\in\Cm$ be the angle such that
\[
\Re\va_z\in[0,\pi],\qquad\cos\va_z=z,\qquad z\in\Cm.
\]
When $\Im{z}=0$ and we can write $z=\nu\in\Rm$, we have (\ref{defvarphinu}). 
In particular for $\nu\in(-1,1)$, we obtain
\[
g(-\nu,\va_{-\nu})=g(-\nu,\pi-\va_{\nu})=g(\nu,2\pi-\va_{\nu})
=g(\nu,\va_{\nu}).
\]
In the case that $\Re{z}=0$, we have
\[
\va_z=\frac{\pi}{2}-i\sinh^{-1}(\Im{z})
=\frac{\pi}{2}-i\ln\left(\Im{z}+\sqrt{(\Im{z})^2+1}\right).
\]
Suppose that $\Im{z}\neq0$ and $\Re{z}\neq0$. We obtain
\ba
\va_z
&=&
\tan^{-1}\left(\mathrm{sgn}(\Re{z})\sqrt{
\frac{1+|z|^2-r}{1-|z|^2+r}}\right)
\\
&+&
i\ln\left(\frac{\sqrt{|z|^2+1+r}-\mathrm{sgn}(\Im{z})\sqrt{|z|^2-1+r}}
{\sqrt{2}}\right),
\ea
where
\[
r=\sqrt{(|z|^2+1+2\Re{z})(|z|^2+1-2\Re{z})}.
\]

Let $\epsilon$ be an infinitesimally small positive number. For 
$\nu\in(-1,1)$ we have
\ba
\Lambda^{\pm}(\nu)
&:=&
\Lambda(\nu\pm i\epsilon)
\\
&=&
\lambda(\nu)\pm i\frac{\varpi\nu}{2\sqrt{1-\nu^2}}\left[
g(\nu,\va_{\nu})+g(\nu,2\pi-\va_{\nu})\right]
\\
&=&
\lambda(\nu)\pm i\frac{\varpi\nu}{\sqrt{1-\nu^2}}g(\nu,\va_{\nu}).
\ea
We obtain
\ba
\Lambda^+(\nu)-\Lambda^-(\nu)
&=&
\frac{2i\varpi\nu}{\sqrt{1-\nu^2}}g(\nu,\va_{\nu})
\\
&=&
\frac{2i\varpi\nu}{\sqrt{1-\nu^2}}\left[1+2\sum_{m=1}^L\beta_m
\gamma_m(\nu)\cos\left(m\va_{\nu}\right)\right],
\ea
and
\[
\Lambda^+(\nu)\Lambda^-(\nu)
=\lambda(\nu)^2+\frac{\varpi^2\nu^2}{1-\nu^2}g(\nu,\va_{\nu})^2.
\]

We can estimate the number of discrete eigenvalues as follows.

\begin{prop}
Suppose $\Lambda^+(\nu)\Lambda^-(\nu)\neq0$ for $\nu\in[-1,1]$. 
Then $M\le L+1$.
\end{prop}

\begin{proof}
We prove the statement relying on the argument principle \cite{Mika61}. 
Since $\Lambda(z)$ is holomorphic in the whole plane cut between $-1$ and 
$1$, according to the argument principle, the number of its roots is given by 
\[
2M=\frac{1}{2\pi}\Delta_C\mathop{\mathrm{arg}}\Lambda(z),
\]
where $\Delta_C$ represents the change around the contour $C$ which 
encircles the cut on the real axis from $-1$ to $1$. Due to the assumed 
condition $\Lambda^+(\nu)\Lambda^-(\nu)\neq0$, we have
\[
2M=\frac{1}{2\pi}\left[\Delta_{C_+}\mathop{\mathrm{arg}}\Lambda^+(\nu)
+\Delta_{C_-}\mathop{\mathrm{arg}}\Lambda^-(\nu)\right],
\]
where $\Delta_{C_{\pm}}$ are changes from $1$ to $-1$ and from $-1$ to $1$, 
respectively. Noting that
\[
\Lambda^+(\nu)=\Lambda^-(-\nu),\qquad
\mathop{\mathrm{arg}}\Lambda^+(\nu)=-\mathop{\mathrm{arg}}\Lambda^+(-\nu),
\]
we finally obtain
\[
M=\frac{1}{\pi}\Delta_{0\to1}\mathop{\mathrm{arg}}\Lambda^+(\nu),
\]
where $\Delta_{0\to1}$ is the change when $\nu$ goes from $0$ to $1$. Note that
\[
\mathop{\mathrm{arg}}\Lambda^+(0)=\mathop{\mathrm{arg}}\Lambda^-(0)=0.
\]
We see that $M$ equals one plus the number that $\Lambda^+(\nu)$ crosses the 
real axis as $\nu$ moves from $0$ to $1$, or the number of roots of 
$\Im\Lambda^+(\nu)$. Since $g(\nu,\va_{\nu})$ is an even polynomial of degree 
$2L$, there are at most $L$ zeros on $(0,1)$. Therefore $M\le L+1$.
\end{proof}

\begin{rmk}
In the case of isotropic scattering ($L=0$), $\Lambda(z)$ is obtained as
\[
\Lambda(z)=1-\frac{\varpi z}{\sqrt{z+1}\sqrt{z-1}}.
\]
This $\Lambda(z)$ has only two roots $z=\pm\nu_0$ ($\nu_0>1$), i.e., 
$\Lambda(\pm\nu_0)=0$, and we obtain
\be
\nu_0=\frac{1}{\sqrt{1-\varpi^2}}.
\label{exactnu0}
\ee
Thus the largest eigenvalue $\nu_0$ can be explicitly written down in 
flatland. 
We note that in three dimensions with planar symmetry the largest eigenvalue 
is only obtained as a solution to the following transcendental equation 
\cite{Case60,Case-Zweifel}
\[
1=\varpi\nu_0\tanh^{-1}(1/\nu_0).
\]
\end{rmk}

In the rest of this section, we explore orthogonality relations for 
$\phi(\nu,\va)$.

\begin{lem}
\label{orthlem}
Suppose $\nu_1,\nu_2\in\sigma$ are different, i.e., $\nu_1\neq\nu_2$. Then,
\[
\int_0^{2\pi}\mu\phi(\nu_1,\va)\phi(\nu_2,\va)\,d\va=0.
\]
\end{lem}

\begin{proof}
We consider the following two equations.
\ba
\left(1-\frac{\mu}{\nu_1}\right)\phi(\nu_1,\va)
&=&
\frac{\varpi}{2\pi}+\frac{\varpi}{\pi}\sum_{m=1}^L\beta_m
\gamma_m(\nu_1)\cos{m\va},
\\
\left(1-\frac{\mu}{\nu_2}\right)\phi(\nu_2,\va)
&=&
\frac{\varpi}{2\pi}+\frac{\varpi}{\pi}\sum_{m=1}^L\beta_m
\gamma_m(\nu_2)\cos{m\va}.
\ea
We multiply the upper equation by $\phi(\nu_2,\va)$ and the lower equation 
by $\phi(\nu_1,\va)$, integrate over $\va$, and subtract the second equation 
from the first equation. We obtain
\[
\left(\frac{1}{\nu_2}-\frac{1}{\nu_1}\right)\int_0^{2\pi}\mu
\phi(\nu_1,\va)\phi(\nu_2,\va)\,d\va=0.
\]
Thus the proof is completed.
\end{proof}

\begin{thm}
\label{orth}
Consider $\nu,\nu'\in\sigma$. Let $\mathcal{N}(\nu)$ be the normalization 
factor in (\ref{normalizationfactor}). We have
\[
\int_0^{2\pi}\mu\phi(\nu,\va)\phi(\nu',\va)\,d\va
=\mathcal{N}(\nu)\delta(\nu-\nu').
\]
Here the Dirac delta function $\delta(\nu-\nu')$ is read as the Kronecker 
delta $\delta_{\nu,\nu'}$ for $\nu,\nu'\notin[-1,1]$.
\end{thm}

\begin{proof}
According to Lemma \ref{orthlem}, the integral vanishes for $\nu\neq\nu'$. 
Hence it is enough if we show
\[
\int_0^{2\pi}\mu\phi(\nu,\va)^2\,d\va=\mathcal{N}(\nu).
\]

In the spirit of \cite{McCormick-Kuscer66}, we begin by defining
\[
J(z,z')=\int_0^{2\pi}\mu\frac{g(z,\va)}{z-\mu}\frac{g(z',\va)}{z'-\mu}\,d\va,
\]
for $z,z'\in\Cm\setminus[-1,1]$. We assume $z\neq z'$. We have
\ba
J(z,z')
&=&
\frac{1}{z'-z}\int_0^{2\pi}\mu g(z,\va)g(z',\va)
\left(\frac{1}{z-\mu}-\frac{1}{z'-\mu}\right)\,d\va
\\
&=&
\frac{1}{z'-z}\Biggl[\bar{\Gamma}_0(z)-\bar{\Gamma}_0(z')+
2\sum_{m=1}^L\beta_m
\Bigl(\gamma_m(z')\bar{\Gamma}_m(z)-\gamma_m(z)\bar{\Gamma}_m(z')\Bigr)\Biggr],
\ea
where
\[
\bar{\Gamma}_m(z)=\int_0^{2\pi}\mu\frac{g(z,\va)}{z-\mu}\cos{m\va}\,d\va.
\]
Let us write $\bar{\Gamma}_m(z)$ as
\be
\bar{\Gamma}_m(z)
=\Gamma_m(z)+\int_0^{2\pi}\mu\frac{g(\mu,\va)}{z-\mu}\cos{m\va}\,d\va,
\label{gammabargamma}
\ee
where
\ba
\Gamma_m(z)&=&
\int_0^{2\pi}\mu\frac{g(z,\va)-g(\mu,\va)}{z-\mu}\cos{m\va}\,d\va
\\
&=&
2\sum_{n=1}^L\beta_n\int_0^{2\pi}\mu\frac{\gamma_n(z)-\gamma_n(\mu)}{z-\mu}
\cos{n\va}\cos{m\va}\,d\va.
\ea
The second term of the right-hand side of (\ref{gammabargamma}) is calculated 
as
\ba
&&
\int_0^{2\pi}\mu\frac{g(\mu,\va)}{z-\mu}\cos{m\va}\,d\va
\\
&&=
\frac{1}{2i\varpi}\int_{-1}^1\frac{\Lambda^+(\mu)-\Lambda^-(\mu)}{z-\mu}
\cos(m\va)\,d\mu
\\
&&=
\frac{i}{2\varpi}\oint\frac{\Lambda(w)\cos\left(m\va_w\right)}{z-w}\,dw
\\
&&=
\frac{\pi}{\varpi}\Lambda(z)\cos\left(m\va_z\right),
\ea
where we chose the contour encircling the interval between $-1$ and $1$. 
Hence we have
\[
\bar{\Gamma}_m(z)=\Gamma_m(z)
+\frac{\pi}{\varpi}\Lambda(z)\cos\left(m\va_z\right).
\]
Therefore,
\ba
J(z,z')
&=&
\frac{1}{z'-z}\Biggl[\Gamma_0(z)-\Gamma_0(z')+
2\sum_{m=1}^L\beta_m
\Bigl(\gamma_m(z')\Gamma_m(z)-\gamma_m(z)\Gamma_m(z')\Bigr)\Biggr]
\\
&+&
\frac{\pi}{\varpi}\frac{g(z',\va_z)\Lambda(z)-g(z,\va_{z'})\Lambda(z')}{z'-z}.
\ea
On the right-hand side of the above equation, the first part vanishes. 
The last term can be rewritten as
\ba
\frac{g(z',\va_z)\Lambda(z)-g(z,\va_{z'})\Lambda(z')}{z'-z}
&=&
\frac{g(z',\va_z)-g(z,\va_z)}{z'-z}\Lambda(z)
-\frac{g(z,\va_{z'})-g(z,\va_z)}{z'-z}\Lambda(z')
\\
&-&
g(z,\va_z)\frac{\Lambda(z')-\Lambda(z)}{z'-z}.
\ea
Note that
\ba
\frac{g(z',\va_z)-g(z,\va_z)}{z'-z}
&=&
2\sum_{m=1}^L\beta_m\frac{\gamma_m(z')-\gamma_m(z)}{z'-z}\cos(m\va_z),
\\
\frac{g(z,\va_{z'})-g(z,\va_z)}{z'-z}
&=&
2\sum_{m=1}^L\beta_m\gamma_m(z)
\frac{\cos(m\va_{z'})-\cos(m\va_z)}{z'-z}.
\ea
Let $\nu\notin[-1,1]$ be a discrete eigenvalue. We bring $z$ to $\nu$ and 
then let $z'$ approach $\nu$. We obtain
\ba
\lim_{z'\to\nu}\lim_{z\to\nu}J(z,z')
&=&
-g(\nu,\va_{\nu})\lim_{z'\to\nu}\lim_{z\to\nu}
\frac{\Lambda(z')-\Lambda(z)}{z'-z}
\\
&=&
-g(\nu,\va_{\nu})\frac{d\Lambda(\nu)}{d\nu}.
\ea
That is,
\[
\int_0^{2\pi}\mu\phi(\nu,\va)^2\,d\va
=\left(\frac{\varpi\nu}{2\pi}\right)^2J(\nu,\nu)
=\left(\frac{\varpi\nu}{2\pi}\right)^2
g(\nu,\va_{\nu})\frac{d\Lambda(\nu)}{d\nu}.
\]

Next we suppose that $\nu,\nu'\in(-1,1)$. We need to be careful about 
changing the order of integrals. Using the Poincare-Bertrand formula 
\cite{Muskhelishvili}:
\[
\frac{\mathcal{P}}{\nu-\mu}\frac{\mathcal{P}}{\nu'-\mu}
=\frac{1}{\nu-\nu'}\left(\frac{\mathcal{P}}{\nu'-\mu}
-\frac{\mathcal{P}}{\nu-\mu}\right)
+\pi^2\delta(\nu-\mu)\delta(\nu'-\mu),
\]
we have
\ba
\int_0^{2\pi}\mu\phi(\nu,\va)\phi(\nu',\va)\,d\va
&=&
\left(\frac{\varpi}{2\pi}\right)^2\nu\nu'\int_0^{2\pi}\mu\mathcal{P}
\frac{g(\nu,\va)}{\nu-\mu}\mathcal{P}\frac{g(\nu',\va)}{\nu'-\mu}\,d\va
\\
&+&
\frac{\varpi\nu}{2\pi}\lambda(\nu')\pint_0^{2\pi}\mu\frac{g(\nu,\va)}{\nu-\mu}
\delta(\nu'-\mu)\,d\va
\\
&+&
\frac{\varpi\nu'}{2\pi}\lambda(\nu)\pint_0^{2\pi}\mu
\frac{g(\nu',\va)}{\nu'-\mu}\delta(\nu-\mu)\,d\va
\\
&+&
\lambda(\nu)\lambda(\nu')\int_0^{2\pi}\mu\delta(\nu-\mu)\delta(\nu'-\mu)\,d\va.
\ea
Therefore,
\ba
&&
\int_0^{2\pi}\mu\phi(\nu,\va)^2\,d\va
\\
&&=
\lim_{\nu'\to\nu}\frac{1}{\nu-\nu'}\pint_0^{2\pi}\mu
\left[\frac{\varpi\nu}{2\pi}g(\nu,\va)\phi(\nu',\va)-\frac{\varpi\nu'}{2\pi}
g(\nu',\va)\phi(\nu,\va)\right]\,d\va
\\
&&+
\delta(\nu-\nu')\lim_{\nu'\to\nu}\int_0^{2\pi}\mu
\left[\frac{\varpi^2\nu\nu'}{4}g(\nu,\va)g(\nu',\va)+\lambda(\nu)\lambda(\nu')
\right]\delta(\nu-\mu)\,d\va.
\ea
The first term on the right-hand side vanishes. The second term on the 
right-hand side is computed as
\ba
&&
\lim_{\nu'\to\nu}\int_0^{2\pi}\mu\left[\frac{\varpi^2\nu\nu'}{4}g(\nu,\va)
g(\nu',\va)+\lambda(\nu)\lambda(\nu')\right]\delta(\nu-\mu)\,d\va
\\
&&
=\left[\left(\frac{\varpi\nu}{2}\right)^2g(\nu,\va_{\nu})^2+\lambda(\nu)^2
\right]\frac{2\nu}{\sqrt{1-\nu^2}}.
\ea
\end{proof}

\section{One-dimensional Green's function}
\label{1dGreen}

Let us consider the Green's function $G(x,\va;\va_0)$ which satisfies
\[
\left(\mu\frac{\pp}{\pp x}+1\right)G(x,\va;\va_0)=
\varpi\int_0^{2\pi}p(\va,\va')G(x,\va';\va_0)\,d\va'+
\delta(x)\delta(\va-\va_0),
\]
and $G(x,\va;\va_0)\to0$ as $|x|\to\infty$.  The completeness of singular 
eigenfunctions can be shown in the usual way \cite{Mika61}. The Green's 
function is given by
\[
\left\{\begin{aligned}
G(x,\va;\va_0)=
\sum_{j=0}^{M-1}A_{j+}\psi_{\nu_j}(x,\va)
+\int_0^1A(\nu)\psi_{\nu}(x,\va)\,d\nu,
&\quad x>0,
\\
G(x,\va;\va_0)=
-\sum_{j=0}^{M-1}A_{j-}\psi_{-\nu_j}(x,\va)
-\int_{-1}^0A(\nu)\psi_{\nu}(x,\va)\,d\nu,
&\quad x<0,
\end{aligned}
\right.
\]
with some coefficients $A_{j\pm}$ ($j=0,\dots,M-1$) and $A(\nu)$. 
The jump condition is written as
\[
G(0^+,\va;\va_0)-G(0^-,\va;\va_0)=\frac{1}{\mu}\delta(\va-\va_0).
\]
Hence we have
\[
\sum_{j=0}^{M-1}\Bigl[A_{j+}\phi(\nu_j,\va)+A_{j-}\phi(-\nu_j,\va)\Bigr]
+\int_{-1}^1A(\nu)\phi(\nu,\va)\,d\nu=
\frac{1}{\mu}\delta(\va-\va_0).
\]
Using orthogonality relations given in Theorem \ref{orth}, the coefficients 
$A_{j\pm}$, $A(\nu)$ are determined as
\[
A_{j\pm}=\frac{\phi(\pm\nu_j,\va_0)}{\mathcal{N}(\pm\nu_j)},\qquad
A(\nu)=\frac{\phi(\nu,\va_0)}{\mathcal{N}(\nu)}.
\]
Therefore we obtain the one-dimensional Green's function as
\[
G(x,\va;\va_0)=
\sum_{j=0}^{M-1}
\frac{\phi(\pm\nu_j,\va_0)\phi(\pm\nu_j,\va)}{\mathcal{N}(\nu_j)}
e^{-|x|/\nu_j}+
\int_0^1\frac{\phi(\pm\nu,\va_0)\phi(\pm\nu,\va)}{\mathcal{N}(\nu)}
e^{-|x|/\nu}\,d\nu,
\]
where upper signs are used for $x>0$ and lower signs are used for $x<0$.

\section{Two-dimensional transport theory in flatland}
\label{2d}

To find the Green's function in (\ref{rte2d}), we consider the following 
homogeneous equation,
\be
\left(\uv\cdot\nabla+1\right)\psi(\vv{\rho},\va)
=\varpi\int_0^{2\pi}p(\va,\va')\psi(\vv{\rho},\va')\,d\va'.
\label{rtehomo}
\ee
We consider rotation of the reference frame for some unit vector 
$\uvk\in\Cm^2$ ($\uvk\cdot\uvk=1$). By an operator $\rrf{\uvk}$, we measure 
angles in the reference frame whose $x$-axis lies in the direction of $\uvk$. 
We have
\[
\rrf{\uvk}\va=\va-\va_{\uvk},
\]
where $\va_{\uvk}$ is the angle of $\uvk$ in the laboratory frame. The dot 
product $\uv\cdot\uvk$ is expressed as
\[
\uv\cdot\uvk=\rrf{\uvk}\mu.
\]
\mm{We find that the inverse is given by}
\[
\mm{\irrf{\uvk}\va=\va+\va_{\uvk}.}
\]
Let us assume the angular flux $\psi(\vv{\rho},\va)$ has the form
\be
\psi_{\nu}(\vv{\rho},\va)=\rrf{\uvk}\phi(\nu,\va)e^{-\uvk\cdot\vv{\rho}/\nu},
\label{sep}
\ee
where $\nu$ is the separation constant. We will see that this $\phi(\nu,\va)$ 
is the singular eigenfunction developed in Section \ref{1d}.

By plugging (\ref{sep}) into (\ref{rtehomo}) we obtain
\be
\left(1-\frac{\uv\cdot\uvk}{\nu}\right)\rrf{\uvk}\phi(\nu,\va)=
\varpi\int_0^{2\pi}\left[\frac{1}{2\pi}+\frac{1}{\pi}\sum_{m=1}^L\beta_m
\cos\Bigl(m(\rrf{\uvk}\va-\rrf{\uvk}\va')\Bigr)\right]
\rrf{\uvk}\phi(\nu,\va')\,d\va',
\label{plugged}
\ee
where we used $\va-\va'=\rrf{\uvk}\va-\rrf{\uvk}\va'$. By inverse rotation 
$\irrf{\uvk}$, (\ref{plugged}) reduces to (\ref{phieq}). That is, 
$\rrf{\uvk}\phi(\nu,\va)$ is the singular eigenfunction 
\mm{for $\nu\in\sigma$} measured in the reference frame which is rotated by 
$\va_{\uvk}$.

\mm{The unit vector $\uvk$ is written as}
\[
\mm{\uvk=\left(\begin{array}{c}\cos{\va_{\uvk}}\\\sin{\va_{\uvk}}
\end{array}\right),}
\]
\mm{with angle $\va_{\uvk}$. Let us set $\uvk$ as}
\[
\uvk=\uvk(\nu q)
=\left(\begin{array}{c}\hat{k}_x(\nu q)\\-i\nu q\end{array}\right),
\]
where $q\in\Rm$ and
\[
\hat{k}_x(\nu q)=\sqrt{1+(\nu q)^2}.
\]
We will show orthogonality relations for $\rrf{\uvk}\phi(\nu,\va)$. 

\begin{thm}
\label{orth2d}
For $\nu,\nu'\in\sigma$ and any $q\in\Rm$ we have
\[
\int_0^{2\pi}\mu\left[\rrf{\uvk(\nu q)}\phi(\nu,\mu)\right]
\left[\rrf{\uvk(\nu' q)}\phi(\nu',\mu)\right]\,d\va=
\hat{k}_x(\nu q)\mathcal{N}(\nu)\delta(\nu-\nu'),
\]
\end{thm}

\begin{proof}
Similar to (\ref{plugged}) we have
\[
\left(1-\frac{\uv\cdot\uvk}{\nu}\right)\rrf{\uvk}\phi(\nu,\va)=
\frac{\varpi}{2\pi}\int_0^{2\pi}\sum_{m=-L}^L\beta_me^{im(\va-\va')}
\rrf{\uvk}\phi(\nu,\va')\,d\va'.
\]
For $\uvk_1=\uvk(\nu_1q)$ and $\uvk_2=\uvk(\nu_2q)$ with a fixed $q$, 
we have
\ba
\left[\rrf{\uvk_2}\phi(\nu_2,\va)\right]
\rrf{\uvk_1}\left(1-\frac{\mu}{\nu_1}\right)\phi(\nu_1,\va)
&=&
\frac{\varpi}{2\pi}\sum_{m=-L}^L\beta_me^{im\va}
\left[\rrf{\uvk_2}\phi(\nu_2,\va)\right]
\\
&\times&
\int_0^{2\pi}e^{-im\va'}\rrf{\uvk_1}\phi(\nu_1,\va')\,d\va',
\\
\left[\rrf{\uvk_1}\phi(\nu_1,\va)\right]
\rrf{\uvk_2}\left(1-\frac{\mu}{\nu_2}\right)\phi(\nu_2,\va)
&=&
\frac{\varpi}{2\pi}\sum_{m=-L}^L\beta_me^{-im\va}
\left[\rrf{\uvk_1}\phi(\nu_1,\va)\right]
\\
&\times&
\int_0^{2\pi}e^{im\va'}\rrf{\uvk_2}\phi(\nu_2,\va')\,d\va',
\ea
where we used $\sum_m\beta_me^{im(\va-\va')}=\sum_m\beta_me^{-im(\va-\va')}$ 
($\beta_{-m}=\beta_m$). By subtraction and integration from $0$ to $2\pi$, 
we obtain
\[
\int_0^{2\pi}
\left(\frac{\rrf{\uvk_2}\mu}{\nu_2}-\frac{\rrf{\uvk_1}\mu}{\nu_1}\right)
\left[\rrf{\uvk_1}\phi(\nu_1,\va)\right]
\left[\rrf{\uvk_2}\phi(\nu_2,\va)\right]\,d\va
=0.
\]
We note that
\[
\rrf{\uvk}\mu=\uv\cdot\uvk=\hat{k}_x(\nu q)\cos{\va}-i\nu q\sin{\va}.
\]
Thus,
\[
\left(\frac{\hat{k}_x(\nu_2q)}{\nu_2}-\frac{\hat{k}_x(\nu_1q)}{\nu_1}\right)
\int_0^{2\pi}\cos{\va}\left[\rrf{\uvk_1}\phi(\nu_1,\va)\right]
\left[\rrf{\uvk_2}\phi(\nu_2,\va)\right]\,d\va
=0.
\]
We obtain
\[
\int_0^{2\pi}\mu\left[\rrf{\uvk_1}\phi(\nu_1,\va)\right]
\left[\rrf{\uvk_2}\phi(\nu_2,\va)\right]\,d\va=0,\qquad \nu_1\neq\nu_2.
\]
When $\nu=\nu_1=\nu_2$, we can calculate the integral as
\ba
\int_0^{2\pi}\mu\left[\rrf{\uvk}\phi(\nu,\va)\right]^2\,d\va
&=&
\int_0^{2\pi}\left[\irrf{\uvk}\mu\right]\phi(\nu,\va)^2\,d\va
\\
&=&
\hat{k}_x(\nu q)\int_0^{2\pi}\mu\phi(\nu,\va)^2\,d\va
\\
&=&
\hat{k}_x(\nu q)\mathcal{N}(\nu),
\ea
where we used $\irrf{\uvk}\mu\mm{=\cos(\irrf{\uvk}\va)=\cos(\va+\va_{\uvk})}
=\hat{k}_x(\nu q)\cos{\va}+i\nu q\sin{\va}$. Thus the orthogonality relations 
are proved.
\end{proof}

\section{Two-dimensional Green's function}
\label{2dGreen}

Let us consider the radiative transport equation (\ref{rte2d}). We can write 
the jump condition as
\[
G(0^+,y,\va;\va_0)-G(0^-,y,\va;\va_0)=\frac{1}{\mu}\delta(y)\delta(\va-\va_0).
\]
With the completeness of $\phi(\nu,\va)$ and plane-wave modes, the Green's 
function can be written as a superposition of $\psi_{\nu}(\vv{\rho},\va)$ in 
(\ref{sep}). Depending on $x$ we can write
\[
\left\{\begin{aligned}
G(\vv{\rho},\va;\va_0)=
\int_{-\infty}^{\infty}\Biggl[\sum_{j=0}^{M-1}A_{j+}(q)
\psi_{\nu_j}(\vv{\rho},\va)
+\int_0^1A(\nu,q)\psi_{\nu}(\vv{\rho},\va)\,d\nu\Biggr]\,\frac{dq}{2\pi},
&\quad x>0,
\\
G(\vv{\rho},\va;\va_0)=
-\int_{-\infty}^{\infty}\sum_{j=0}^{M-1}\Biggl[
A_{j-}(q)\psi_{-\nu_j}(\vv{\rho},\va)
+\int_{-1}^0A(\nu,q)\psi_{\nu}(\vv{\rho},\va)\,d\nu\Biggr]\,\frac{dq}{2\pi},
&\quad x<0.
\end{aligned}
\right.
\]
Here $A_{j\pm}(q)$, $A(\nu,q)$ are some coefficients. The jump condition reads
\ba
&&
\int_{-\infty}^{\infty}e^{iqy}\Biggl[\sum_{j=0}^{M-1}\Bigl(A_{j+}(q)
\rrf{\uvk(\nu_jq)}\phi(\nu_j,\va)+
A_{j-}(q)\rrf{\uvk(-\nu_jq)}\phi(-\nu_j,\va)\Bigr)
\\
&&+
\int_{-1}^1A(\nu,q)\rrf{\uvk(\nu q)}\phi(\nu,\va)\,d\nu\Biggr]
\,\frac{dq}{2\pi}
\\
&&=
\frac{1}{\mu}\delta(y)\delta(\va-\va_0).
\ea
By using Theorem \ref{orth2d} of orthogonality relations, the coefficients 
$A_{j\pm}(q)$, $A(\nu,q)$ are determined as
\[
A_{j\pm}(q)=\frac{\rrf{\uvk(\pm\nu_jq)}\phi(\pm\nu_j,\va_0)}
{\hat{k}_x(\nu_jq)\mathcal{N}(\pm\nu_j)},\quad
A(\nu,q)=\frac{\rrf{\uvk(\nu q)}\phi(\nu,\va_0)}
{\hat{k}_x(\nu q)\mathcal{N}(\nu)}.
\]
Therefore the Green's function is obtained as
\begin{widetext}
\ba
G(\vv{\rho},\va;\va_0)
&=&
\frac{1}{2\pi}\int_{-\infty}^{\infty}e^{iqy}\Biggl[
\sum_{j=0}^{M-1}
\frac{\rrf{\uvk(\pm\nu_jq)}\phi(\pm\nu_j,\va_0)\phi(\pm\nu_j,\va)}
{\hat{k}_x(\nu_jq)\mathcal{N}(\nu_j)}e^{-\hat{k}_x(\nu_jq)|x|/\nu_j}
\\
&+&
\int_0^1\frac{\rrf{\uvk(\pm\nu q)}\phi(\pm\nu,\va_0)\phi(\pm\nu,\va)}
{\hat{k}_x(\nu q)\mathcal{N}(\nu)}e^{-\hat{k}_x(\nu q)|x|/\nu}\,d\nu
\Biggr]\,dq,
\ea
where upper signs are used for $x>0$ and lower signs are used for $x<0$. 
The above Green's function can be rewritten as (\ref{mainresult}).
\end{widetext}

\section{Energy density}
\label{ene}

\mm{Let us calculate the energy density $u$ for an isotropic source 
$\delta(\vv{\rho})$, i.e.,}
\[
\mm{u=\int_0^{2\pi}\int_0^{2\pi}G(\vv{\rho},\va;\va_0)\,d\va d\va_0.}
\]
\mm{Without loss of generality, we can put $y=0$ and assume $x>0$. Using 
(\ref{mainresult}), we obtain}
\[
\mm{u(x)=\frac{1}{2\pi}\int_{-\infty}^{\infty}\left[\sum_{j=0}^{M-1}
\frac{e^{-\hat{k}_x(\nu_jq)x/\nu_j}}{\hat{k}_x(\nu_jq)\mathcal{N}(\nu_j)}
+\int_0^1\frac{e^{-\hat{k}_x(\nu q)x/\nu}}{\hat{k}_x(\nu q)\mathcal{N}(\nu)}
\,d\nu\right]\,dq.}
\]
\mm{We note that}
\[
\mm{\int_{-\infty}^{\infty}\frac{e^{-\hat{k}_x(\nu q)x/\nu}}
{\hat{k}_x(\nu q)}\,dq
=\frac{2}{\nu}\int_0^{\infty}\frac{e^{-\sqrt{1+t^2}x/\nu}}
{\sqrt{1+t^2}}\,dt=
\frac{2}{\nu}\int_1^{\infty}\frac{e^{-sx/\nu}}{\sqrt{s^2-1}}\,ds=
\frac{2}{\nu}K_0\left(\frac{x}{\nu}\right),}
\]
\mm{where $K_0$ is the modified Bessel function of the second kind of order 
zero. Hence we have}
\be
\mm{u(x)=\frac{1}{\pi}\left[\sum_{j=0}^{M-1}\frac{K_0(x/\nu_j)}
{\nu_j\mathcal{N}(\nu_j)}+\int_0^1\frac{K_0(x/\nu)}{\nu\mathcal{N}(\nu)}\,d\nu
\right].}
\label{enedensityresult}
\ee
\mm{The expression (\ref{enedensityresult}) is the general result. Let us 
consider the case of isotropic scattering ($L=0$). Since $g(\nu,\va)=1$, 
we have}
\ba
\mm{\Lambda(z)}&\mm{=}&
\mm{1-\frac{\varpi z}{2\pi}\int_0^{2\pi}\frac{1}{z-\mu}\,d\va=
1-\frac{\varpi z}{\sqrt{z^2-1}},\quad z\in\Cm\setminus[-1,1],}
\\
\mm{\lambda(\nu)}&\mm{=}&
\mm{1-\frac{\varpi\nu}{2\pi}\pint_0^{2\pi}\frac{1}{\nu-\mu}\,d\va=
1-\frac{\varpi|\nu|}{\pi\sqrt{1-\nu^2}}\left(
\ln\left|\frac{\nu}{1+\sqrt{1-\nu^2}}\right|+\cosh^{-1}\frac{1}{\nu}\right),}
\ea
\mm{where $\nu\in(-1,1)$. Using the above $\Lambda(z)$ and $\lambda(\nu)$, we 
can calculate $\mathcal{N}(\nu)$ in (\ref{normalizationfactor}). We note that 
$M=1$ and the positive root $\nu_0$ such that $\Lambda(\nu_0)=0$ is given 
in (\ref{exactnu0}).}

\section{Concluding remarks}
\label{conclusions}

We have obtained the Green's function for the radiative transport equation in 
flatland with separation of variables. As an alternative way, the Green's 
function can also be found with the Fourier transform. This calculation is 
summarized in Appendix. \mm{Assuming the completeness of singular 
eigenfunctions in the presence of boundaries, the Green's function is given 
as a superposition of singular eigenfunctions, and the coefficients 
$A_{j\pm}(q),A(\nu,q)$ in Sec.~\ref{2dGreen} are determined from the boundary 
conditions. If we consider the present extension of Case's method, i.e., the 
separation of variables developed in this paper, in the planar geometry, we 
need the half-range completeness to express the Green's function as a 
superposition of singular eigenfunctions. Moreover we need to establish 
the half-range orthogonality relations to calculate $A_{j\pm}(q),A(\nu,q)$. 
It is another important future work to develop the separation of variables 
for the time-dependent radiative transport equation.}


\appendix

\section{Fourier transform}

We will find an alternative expression of the Green's function 
(\ref{firstalternative}) by using the Fourier transform 
\cite{Ganapol00,Ganapol15}.

Let us introduce the Fourier transform as
\[
\tilde{G}(\vv{k},\va;\va_0)=
\int_{\Rm^2}e^{-i\vv{k}\cdot\vv{\rho}}G(\vv{\rho},\va;\va_0)\,d\vv{\rho}.
\]
By introducing
\be
\tilde{G}_m(\vv{k})=\int_0^{2\pi}\left[\rrf{\uvk}e^{-im\va}\right]
\tilde{G}(\vv{k},\va;\va_0)\,d\va,
\label{defGm}
\ee
we can write (\ref{rte2d}) as
\be
\left(1+i\vv{k}\cdot\uv\right)\tilde{G}(\vv{k},\va;\va_0)=
\frac{\varpi}{2\pi}\sum_{m=-L}^L\beta_m\left[\rrf{\uvk}e^{im\va}\right]
\tilde{G}_m(\vv{k})+
\delta(\va-\va_0).
\label{tildeGeq}
\ee
Note that $P_m(z)$ in (\ref{bigPdef}) can be written as
\[
P_m(z)=\frac{1}{2\pi}\int_0^{2\pi}\frac{e^{im\va}}{z-\mu}\,d\va.
\]
Hereafter we set
\[
z=\frac{i}{k}.
\]
We then have
\be
\tilde{G}_j(\vv{k})=
\varpi z\sum_{m=-L}^L\beta_mP_{m-j}(z)\tilde{G}_m(\vv{k})+
\frac{z}{z-\uvk\cdot\uv_0}\rrf{\uvk}e^{-ij\va_0},\qquad |j|\le L.
\label{tildeGj}
\ee
Hence
\[
\sum_{m=-L}^L\left[\delta_{jm}-\varpi z\beta_mP_{m-j}(z)\right]
\tilde{G}_m(\vv{k})
=\frac{z}{z-\uvk\cdot\uv_0}\rrf{\uvk}e^{-ij\va_0},
\]
for $|j|\le L$. 
Thus, by using (\ref{tildeGeq}), $\tilde{G}(\vv{k},\va;\va_0)$ can be 
expressed using matrices as
\bea
\tilde{G}(\vv{k},\va;\va_0)
&=&
\frac{z}{z-\uvk\cdot\uv_0}\delta(\va-\va_0)
+\frac{\varpi}{2\pi}\frac{z}{z-\uvk\cdot\uv}\frac{z}{z-\uvk\cdot\uv_0}
\nonumber \\
&\times&
\sum_{m=-L}^L\vv{P}^{\dagger}(\uvk,\va)\vv{W}
\left[\vv{I}-\varpi z\vv{L}(z)\vv{W}\right]^{-1}\vv{P}(\uvk,\va_0).
\nonumber \\
\label{Gtildek}
\eea
Here,
\[
\{\vv{L}(z)\}_{jm}=P_{m-j}(z),\qquad
\{\vv{W}\}_{jl}=\beta_m\delta_{jm},\qquad
\{\vv{P}(\uvk,\va)\}_m=e^{-im(\va-\va_{\uvk})}.
\]
We note that
\[
\int_{\Rm^2}\frac{e^{-\rho}}{\rho}\delta(\va_0-\va_{\hvv{\rho}})
e^{-i\vv{k}\cdot\vv{\rho}}\,d\vv{\rho}
=\frac{1}{1+i\vv{k}\cdot\uv_0}.
\]
Therefore we obtain the first alternative expression:
\bea
G(\vv{\rho},\va;\va_0)
&=&
\frac{e^{-\rho}}{\rho}\delta(\va_0-\va_{\hvv{\rho}})\delta(\va-\va_0)
\nonumber \\
&+&
\frac{\varpi}{(2\pi)^3}\int_{\Rm^2}e^{i\vv{k}\cdot\vv{\rho}}
\frac{M(\vv{k},\va,\va_0)}{(1+i\vv{k}\cdot\uv)(1+i\vv{k}\cdot\uv_0)}\,d\vv{k},
\nonumber \\
\label{firstalternative}
\eea
where
\[
M(\vv{k},\va,\va_0)=
\sum_{m=-L}^L\vv{P}^{\dagger}(\uvk,\va)\vv{W}
\left[\vv{I}-\varpi z\vv{L}(z)\vv{W}\right]^{-1}\vv{P}(\uvk,\va_0).
\]



\end{document}